\documentclass[letterpaper, 10pt, conference]{ieeeconf}  

\usepackage{graphicx,wrapfig}
\usepackage[utf8]{inputenc}
\usepackage[T1]{fontenc}
\usepackage{mathtools}
\usepackage[english]{babel}
\usepackage{comment}
\usepackage{amsmath}
\usepackage{amssymb}
\usepackage{bm}
\usepackage{psfrag}
\usepackage{esdiff}
\usepackage{float}
\usepackage{caption}
\usepackage{subcaption}
\usepackage{bm}
\usepackage[ruled,vlined,linesnumbered]{algorithm2e}
\usepackage{array}
\usepackage{blindtext}
\usepackage{epstopdf}
\usepackage{ragged2e}
\usepackage{mathrsfs}
\usepackage{tabu}
\graphicspath{ {arxiv_figures/} }

\usepackage{nomencl}
\usepackage{cite}
\usepackage{lmodern}
\usepackage{enumerate}
\usepackage{tikz}
\usepackage{hyperref}
\usepackage{listings}

\IEEEoverridecommandlockouts                              

\DeclareMathOperator*{\argmin}{arg\,min}

\newcommand{\be}{\begin{equation}}
\newcommand{\ee}{\end{equation}}
\newcommand{\bse}{\begin{subequations}}
	\newcommand{\ese}{\end{subequations}}
\newcommand{\bewn}{\begin{equation*}}
\newcommand{\eewn}{\end{equation*}}
\newcommand{\bbmat}{\begin{bmatrix}} 
	\newcommand{\ebmat}{\end{bmatrix}}
\newcommand{\bd}{\begin{displaymath}}
\newcommand{\ed}{\end{displaymath}}
\newcommand{\bea}{\begin{eqnarray}}
\newcommand{\eea}{\end{eqnarray}}
\newcommand{\ba}{\begin{array}}
	\newcommand{\ea}{\end{array}}
\newcommand{\baa}{\begin{array}{ll}}
	\newcommand{\eaa}{\end{array}}
\newcommand{\ds}{\displaystyle}
\newcommand{\bc}{\begin{center}}
	\newcommand{\ec}{\end{center}}
\newcommand{\ben}{\begin{enumerate}}
	\newcommand{\een}{\end{enumerate}}
\newcommand{\bi}{\begin{itemize}}
	\newcommand{\ei}{\end{itemize}}
\newcommand{\bt}{\begin{tabular}}
	\newcommand{\et}{\end{tabular}}
\newcommand{\bte}{\begin{table}}
	\newcommand{\ete}{\end{table}}
	
\renewcommand{\baselinestretch}{0.987}	

\newcommand{\norm}[1]{\left\lVert#1\right\rVert}   

\makeatletter
\renewcommand\paragraph{\@startsection{paragraph}{4}{\z@}%
	{-2.5ex\@plus -1ex \@minus -.25ex}%
	{1.25ex \@plus .25ex}%
	{\normalfont\normalsize\bfseries}}
\makeatother



\newtheorem{proposition}{Proposition}
\newtheorem{theorem}{Theorem}

\newtheorem{lemma}[theorem]{\textbf{Lemma}}
\newtheorem{assumption}{\textbf{Assumption}}
\newtheorem{problem}{\textbf{Problem}}
\newtheorem{definition}{\textbf{Definition}}
\newtheorem{condition}{\textbf{Condition}}

\newcommand{\bR}{\mathbb{R}}
\newcommand{\bS}{\mathbb{S}}
\newcommand{\bC}{\mathbb{C}}
\newcommand{\calA}{\mathcal{A}}
\newcommand{\calB}{\mathcal{B}}

\newcommand{\calD}{\mathcal{D}}
\newcommand{\calE}{\mathcal{E}}
\newcommand{\calF}{\mathcal{F}}
\newcommand{\calG}{\mathcal{G}}

\newcommand{\calP}{\mathcal{P}}

\newcommand{\calR}{\mathcal{R}}
\newcommand{\calT}{\mathcal{T}}
\newcommand{\calS}{\mathcal{S}}

\newcommand{\calV}{\mathcal{V}}

\newcommand{\calZ}{\mathcal{Z}}

\newcommand{\scriptF}{\mathscr{F}}

\newcommand{\pdydx}[2]{\frac{\partial{#1}}{\partial{#2}}}
\newcommand{\doublepdydx}[2]{\frac{\partial^2{#1}}{\partial{{#2}^2}}}

\newcommand{\abs}[1]{\left |#1\right |}

\title{\LARGE \bf
	Herding an Adversarial Swarm in Three-dimensional Spaces
}

\author{Weifan Zhang, Vishnu S. Chipade and Dimitra Panagou
	\thanks{The authors are with the Department of Aerospace Engineering,
		University of Michigan, Ann Arbor, MI, USA;
		{\tt\small (weifanz,vishnuc,dpanagou)@umich.edu}}
	\thanks{This work has been funded by the Center for Unmanned Aircraft Systems (C-UAS), a National Science Foundation Industry/University Cooperative Research Center (I/UCRC) under NSF Award No. 1738714 along with significant contributions from C-UAS industry members.}
}

\begin{document}
	\maketitle
	\thispagestyle{empty}
	\pagestyle{empty}
	
	\begin{abstract}
		
		This paper presents a defense approach to safeguard a protected area against an attack by a swarm of adversarial agents in three-dimensional (3D) space. We extend our 2D `StringNet Herding' approach, in which a closed formation of string-barriers is established around the adversarial swarm to confine their motion and herd them to a safe area, to 3D spaces by introducing 3D-StringNet. 3D-StringNet is a closed 3D formation of triangular net-like barriers. We provide a systematic approach to generate three types of 3D formations that are used in the 3D herding process and modifications to the finite-time convergent control laws developed in our earlier work. Furthermore, for given initial positions of the defenders, we provide conditions on the initial positions of the attackers for which the defenders are guaranteed to gather as a specified formation at a position on the shortest path of the attackers to the protected area before attackers reach there. The approach is investigated in simulations.
	\end{abstract}
	
	\section{Introduction}
	
	A swarm of multiple robots can in principle perform certain tasks more effectively than one individual robot \cite{bayindir2016review}. However, the fast advancement of swarm technology raises concerns with respect to safety. For instance, autonomous robots in the proximity of protected area (e.g., safety-critical infrastructure) may in some cases be considered as a threat (e.g., aerial robots close to airports or stadiums). In our prior work \cite{chipade2019swarmherding,chipade2020swarmherding}, we developed a method called 'StringNet Herding' in which a group of defending agents (defenders) herds the adversarial swarm away from the protected area by enclosing it in a closed formation of string-like barriers, called StringNet. We assumed that the agents of the adversarial swarm (attackers) are risk-averse and tend to move away from the 2D StringNet formation formed by defending agents, and that the motion of all the agents is constrained to a plane of a fixed altitude. However, in practice, the motion of an attacking aerial swarm does not have to be restricted to a plane. Therefore, in this paper, we extend the StringNet approach to 3D environments.
	
	\subsubsection{Related work} 
	Earlier methods in the literature, namely: n-wavefront herding \cite{paranjape2018robotic}, potential field approach \cite{pierson2015bio}, potential cage approach \cite{varava2017herding}, switched system approach \cite{licitra2017single} that are cited in \cite{chipade2020swarmherding} also provide extensions to 3D environments or some hint to extend the presented 2D laws to 3D environments. However, the 3D extensions are limiting due to: 1) dependence on knowing the model of the attackers' motion, 2) lack of modeling of the attackers' intent to reach or attack a certain protected area, 3) simplified motion and environment models. 
	
	In \cite{kim2018three}, a group of aerial robots tows a capture net to herd a maneuvering UAV in a 3D environment. It is proved that the 3D team is able to capture its target in a finite time. 
	However, the capture net is an open surface in 3D space, so the target UAV still has a chance to escape during the herding process. 
	
	\subsubsection{Overview}
	In this paper, we build on the 2D StringNet herding approach \cite{chipade2020swarmherding} under the similar assumption of risk-averse adversarial attackers, i.e., attackers that adjust their course to avoid obstacles. We propose an approach for 3D-StringNet herding, where 3D-StringNet is a formation of expandable, triangular net-like barriers formed by a group of defenders (Fig.~\ref{fig:3D_StringNet_formation}).  Similar to 2D-Stringnet herding, 3D-StringNet herding also consists of four phases: 1) gathering, 2) seeking, 3) enclosing and 4) herding. We design three 3D formations of the defenders namely planar, hemispherical and spherical that are required to be achieved in the phases discussed above in order to effectively enclose the attackers and herd them to a safe area. The control laws designed in \cite{chipade2020swarmherding} are extended to 3D spaces by considering 3D rigid body dynamics. The `3D-StringNet Herding' thus addresses the aforementioned issues similar to its 2D equivalent. We also provide conditions on the initial positions of the attackers for which the defenders are able to achieve a specified formation at a point on the expected path (shortest path to the protected area) of the attackers before the attackers could reach that point. We provide a convex optimization formulation to quickly find these conditions for a given direction from which the attackers are approaching.
	
	In summary, the design of three 3D formations, appropriate modifications to the 2D herding control laws \cite{chipade2020swarmherding}, and the conditions on the initial positions of the attackers for defenders' guaranteed gathering are the main contributions of this paper compared to our previous work.

	\subsubsection{Structure of the paper}Section \ref{sec:preliminary} describes the mathematical modeling and the problems studied. The details of the 3D herding formations are discussed in Section \ref{sec:formation}, while the modifications to the 2D herding approach are provided in Section \ref{sec:herding}. Conditions on the attackers' initial positions for guaranteed gathering are provided in Section \ref{sec:dominance_region}. Simulation results and conclusions are reported in Section \ref{sec:simulations} and \ref{sec:conclusions}.
	
	\section{Modeling and Problem Statement}\label{sec:preliminary}
	\textit{Notation}:  Euclidean norm is denoted by $\norm{\cdot}$. Absolute value is denoted by $\abs{\cdot}$. $\calB_{\rho}(\mathbf{r}_{c})=\{\mathbf{r} \in \bR^3| \norm{\mathbf{r}-\mathbf{r}_{c}}\le \rho\}$ denotes a ball of radius $\rho>0$ centered at the point $\mathbf{r}_c \in \bR^3$. A saturation function $\bm{\Omega}_{\bar{u}}: \bR^2\rightarrow \bR^2$ is defined as: $\bm{\Omega}_{\bar{u}}(\mathbf{g}) =\min(\norm{\mathbf{g}},\bar{u}) \frac{\mathbf{g}}{\norm{\mathbf{g}}}$. We use characters $g$, $s$, $e$, $h$ as subscripts or superscripts to denote gathering, seeking, enclosing and herding phase, respectively. Characters $sb$, $sn$ used as subscripts denote string barrier and StringNet, respectively. Characters $op$, $cl$ used as superscript denote open and closed, respectively. Similarly, characters $sp$, $hs$, $pl$ used as subscript or superscript denote spherical, hemispherical and planar, respectively.
	
	There are $N_a$ attackers denoted as $\calA_i$, $i \in I_a= \{1,2,...,N_a\}$ and $N_d$ defenders denoted as $\calD_j$, $j \in I_d= \{1,2,...,N_d\}$. The protected area $\calP \subset  \mathbb{R}^3$ is defined as $\calP=\{\textbf{r} \in \bR^3 \;| \; \norm{\textbf{r}-\textbf{r}_p}\le \rho_p\}$, and the safe area $\calS \subset  \mathbb{R}^3$ is defined as $\calS=\{\textbf{r} \in \bR^3 \; | \; \norm{\textbf{r}-\textbf{r}_{s}}\le \rho_{s}\}$, where $(\textbf r_p, \rho_p)$ and $(\textbf r_{s}, \rho_{s})$ are the centers and radii of the corresponding areas, respectively. The agents $\calA_i$ and $\calD_j$ are modeled as spheres of radii $\rho_a$ and $\rho_d \le \rho_a$, respectively and move under double integrator (DI) dynamics with quadratic drag: 
	\begin{subequations} \label{eq:attackDyn1}
		\abovedisplayskip=1pt
		\belowdisplayskip=3pt
		\begin{gather}
			\dot{\textbf{r}}_{ai}
			=\textbf{v}_{ai}, \quad \quad 
			\dot{\textbf{v}}_{ai}
			=\textbf{u}_{ai}-C_{D} \norm{\textbf{v}_{ai}}\textbf{v}_{ai};\\
			\dot{\textbf{r}}_{dj}
			=\textbf{v}_{dj}, \quad \quad 
			\dot{\textbf{v}}_{dj}
			=\textbf{u}_{dj}-C_{D} \norm{\textbf{v}_{dj}}\textbf{v}_{dj}; \label{eq:defendDyn1}\\
			\norm{\mathbf{u}_{ai}} \le \bar{u}_a, \quad
			\norm{\mathbf{u}_{dj}} \le \bar{u}_d; 
		\end{gather}
	\end{subequations}
	where $C_D>0$ is the known, constant drag coefficient; $\mathbf{r}_{ai}=[x_{ai}\; y_{ai}\;z_{ai}]^T \in \bR^3$ and $\mathbf{r}_{dj}=[x_{dj}\; y_{dj}, \;z_{dj}]^T \in \bR^3$ are the position vectors of $\calA_i$ and $\calD_j$, respectively; $\mathbf{v}_{ai}=[v_{x_{ai}}\; v_{y_{ai}}\; v_{z_{ai}}]^T \in \bR^3$, $\mathbf{v}_{dj}=[v_{x_{dj}}\; v_{y_{dj}}\;v_{z_{dj}}]^T \in \bR^3$ are the velocity vectors, respectively, and $\mathbf{u}_{ai}=[u_{x_{ai}}\; u_{y_{ai}} \;u_{z_{ai}}]^T \in \bR^3$, $\mathbf{u}_{dj}=[u_{x_{dj}}\; u_{y_{dj}} \;u_{z_{dj}}]^T \in \bR^3$ are the accelerations (the control inputs), respectively. This model poses an inherent speed bound on each agent with limited acceleration control, i.e., $v_{ai}=\norm{\mathbf{v}_{ai}}<\bar{v}_a=\sqrt{\frac{\bar{u}_a}{C_d}}$ and $v_{dj}=\norm{\mathbf{v}_{dj}}<\bar{v}_d=\sqrt{\frac{\bar{u}_d}{C_d}}$. The defenders are assumed to be faster than the attackers, i.e., $\bar{u}_a < \bar{u}_d$ (equivalently $\bar{v}_a < \bar{v}_d$).
	We also assume the following about the information available to the agents.
	\begin{assumption} The defenders have access to the position $\mathbf{r}_{ai}$ and velocity $\mathbf{v}_{ai}$ of the attacker $\calA_i$ that lies inside a circular sensing zone $\calZ_d=\{\mathbf r \in \mathbb{R}^3 |\; \norm{\mathbf{r}-{\mathbf{r}_{pa}}} \le \varrho_d\}$ for all $i \in I_a$, {where $\varrho_d>0$ is the radius of the defenders' sensing zone}. Every attacker $\calA_i$ has a local sensing zone ${\calZ_{ai}}=\{\mathbf{r} \in \bR^3 \;| \; \norm{\mathbf{r}-\mathbf{r}_{ai}}\le \varrho_{ai} \}$, {where $\varrho_{ai}>0$ is the radius of the attacker $\calA_i$'s sensing zone}.
	\end{assumption}
	
	Attackers aim to reach the protected area $\calP$ while the defenders aim to herd the attackers to the safe area $\calS$ before the attackers reach the protected area. Attackers are assumed to stay within a circular connectivity region of radius $\rho_{ac}$ around the attackers' center of mass. To demonstrate the proposed 3D herding approach, we model the motion of the attackers using a leader-follower control strategy \cite{8360424} that uses potential functions, which however is not known to the defenders. We consider the following problems in this paper.
	\begin{problem}
		Design 3D formations of the defenders with minimum number of defenders to enclose the attackers and to herd them to $S$.
	\end{problem}	
	\begin{problem}
		Given the initial positions of the defenders $\mathbf{r}_{dj}(0)$, for all $j \in I_d$, provide conditions on the initial positions $\mathbf{r}_{ai}(0)$, for all $i \in I_a$, of the attackers for which the defenders are able to gather as a specified formation centered at a point on the expected path of the attackers before any attacker reaches the center of the formation.
	\end{problem}
	
	

	\section{3D-StringNet and 3D Formations}\label{sec:formation}
	In this section, we formally define 3D-StringNet and provide a systematic approach to obtain formations of the defenders to generate 3D-StringNets.
	
	\begin{definition}[3D-StringNet] The StringNet $\calG_{sn}= (\calV_{sn},$ $\calE_{sn},\calF_{sn})$ is a graph consisting of: 
		1) the defenders as the vertices, $\calV_{sn}=\{\calD_1,\calD_2,...,\calD_{N_d}\}$;
		2) a set of edges, $\calE_{sn}=\{(\calD_j,\calD_{j'}) \in \calV_{sn}\times \calV_{sn} | \calD_j \overset{s} \longleftrightarrow \calD_{j'} \}$, where $\overset{s} \longleftrightarrow$ denotes an impenetrable and extendable string-barrier between the defenders;
		3) a set of triangular, expandable, net-like barrier faces, $\calF_{sn}=\{(\calD_j,\calD_{j'},\calD_{j''})| \calD_j,\calD_{j'}, \calD_{j''}\in\calV_{sn},\ (\calD_j,\calD_{j'})\in\calE_{sn},\ (\calD_j,\calD_{j''})\in\calE_{sn},\ (\calD_{j'},\calD_{j''})\in\calE_{sn}\}$. The union of the set of faces is a single component, orientable triangle mesh with zero genus, i.e., no holes (Fig.~\ref{fig:3D_StringNet_formation}).
	\end{definition}
	
	A 3D-StringNet is called closed-3D-StringNet when the union of the face set is a closed manifold and we denote the underlying graph as $\calG_{sn}^{cl}=(\calV_{sn}^{cl},$ $\calE_{sn}^{cl},\calF_{sn}^{cl})$ otherwise it is called as open-3D-StringNet and the graph is denoted as $\calG_{sn}^{op}=(\calV_{sn}^{op},$ $\calE_{sn}^{op},\calF_{sn}^{op})$. 
	\begin{figure}[!h]
		\abovedisplayskip -2pt
		\centering
		\includegraphics[width=0.28\textwidth,trim={0cm 2.8cm 0cm 0.25cm},clip]{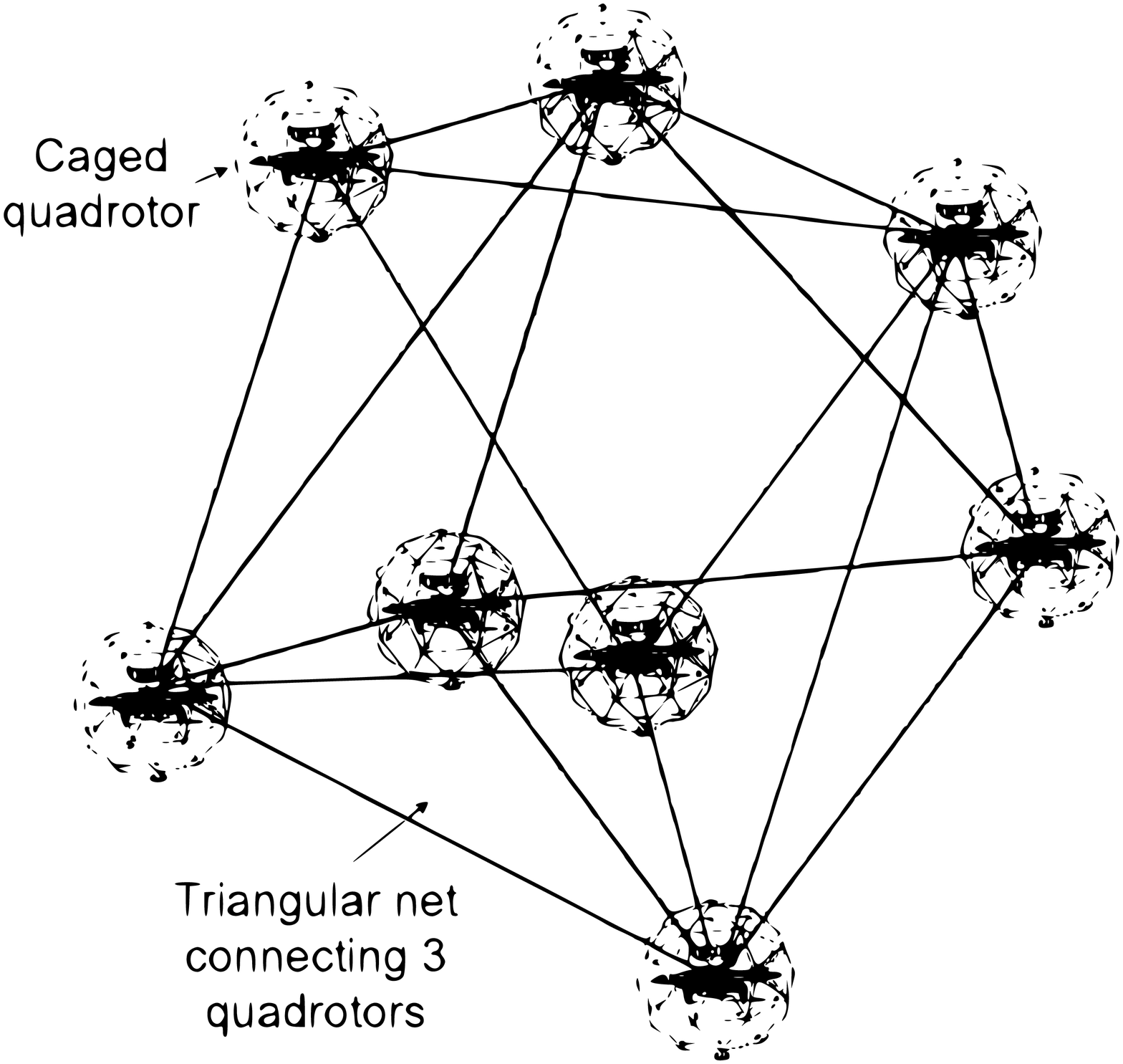}
		\caption{3D StringNet Formation}
		\label{fig:3D_StringNet_formation}
	\end{figure}
	For example, these triangular net-like barriers can look similar to the ones found in \cite{ritz2012cooperative}. We assume that the effect of the triangular net-like barriers on the dynamics of the vehicles is negligible. In practice these triangular net-like barriers can only have finite size. So, we consider the following practical constraints on the edges and the faces in a 3D-StringNet.
	\begin{condition}[Practical Constraint on 3D-StringNet] \label{cond:3D-StringNet}A 3D-StringNet $\calG_{sn}$ should satisfy: $\forall (\calD_j,\calD_k)\in\calE_{sn},\ R_{jk}=||\mathbf{r}_{dj}-\mathbf{r}_{dk}||<\bar{R}_{sb}$, where $\bar{R}_{sb}$ is the maximum length any edge in $\calE_{sn}$ can have.
	\end{condition}
	Condition \ref{cond:3D-StringNet} implies that $\forall (\calD_j,\calD_k,\calD_l)\in\calF_{sn}, \ A_{jkl}^s\leq \frac{\sqrt{3}}{4}(\bar{R}_{sb})^2$, where $A_{jkl}^s$ represents the area of triangular barrier face that is formed by defenders $\calD_j,\;\calD_k,$ and $\calD_l$. 	
	In the next two subsections, we design three 3D formations for the 3D-StringNet that satisfy Condition \ref{cond:3D-StringNet} with the minimum number of defenders required to herd a given a swarm of attackers. 
	\begin{figure}[!h]
		\belowdisplayskip 0pt
		\includegraphics[width=0.48\textwidth,trim={0cm 0cm 0cm 0cm},clip]{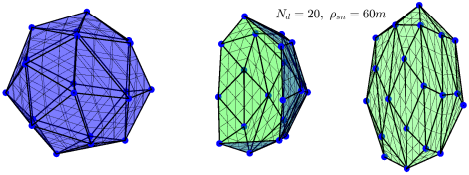}
		\caption{Spherical, hemispherical, and planar formation}
		\vspace{-4mm}
		\label{fig:three_formations}
	\end{figure}
	
	\subsection{Optimal 3D formation for 3D-StringNet Herding}
	
	We want to design a closed 3D-StringNet formation that encloses the connectivity region of the attackers. Since a triangular mesh generated by connecting uniformly distributed points on a sphere contains the largest spatial volume with a given number of points, we choose the locations of the defenders by uniformly distributing them on a sphere. 
	
	The uniform distribution of the defenders on a spherical surface is generated as a solution to the problem of finding the minimum electrostatic potential energy configuration of N electrons constrained on the surface of the unit sphere \cite{koay2014distributing} (Thompson problem). Let $\mathbf{p}_i=[\theta_i,\phi_i]^T$ denote the spherical coordinates of $i^{th}$-electron on the sphere of radius $\rho_{sn}$. The electrostatic potential energy $\Phi_C$ of $N_d$ electrons is expressed as:
	\be \label{eq:eletrostatic_sum} 
	\Phi_C= \textstyle \sum_{i=1}^{N_d}
	\sum_{j\neq i}^{N_d}\frac{1}{\rho_{sn}\sqrt{2(1-\Lambda(\phi_{ij},\theta_i,\theta_j))}},
	\ee
	where $\Lambda(\phi_{ij},\theta_i,\theta_j)=C(\Delta\phi_{ij})S(\theta_i)S(\theta_j)+C(\theta_i)C(\theta_j)$, $S(\theta)=\sin(\theta)$, $C(\theta)=\cos(\theta)$, and $\Delta\phi_{ij}=\phi_i-\phi_j$. $\rho_{sn}$ is the radius of the sphere on which the defenders are distributed. Denote $\mathbf{p}=[\mathbf{p}_1,\mathbf{p}_2,...,\mathbf{p}_{N_d}]^T$. Then, the problem of finding an uniform distribution of electrons is formulated as an unconstrained optimization problem:
	\be \label{eq:min_prob}
	\mathbf{p}^*=\textstyle \argmin_{\mathbf{p}} \quad \Phi_C
	\ee 
	We use gradient flow to find $\mathbf{p}^*$. Starting with some initial locations, the motion of the electrons under gradient flow is governed by:
	$
	\dot{\mathbf{p}}=-\nabla \Phi_C.
	$
	We choose the optimal locations of the electrons in the uniform distribution from \eqref{eq:min_prob} as the desired locations $\bm{\xi}_{l}^{s_0}=\rho_{sn}[\sin(\theta_l^*)\cos(\phi_l^*), \; \sin(\theta_l^*)\sin(\phi_l^*), \; \cos(\theta_l^*)]^T\in \bR^3$, for $l\in I_d$, for the defenders to obtain a closed-3D-StringNet $\calG_{sn}^{cl}$. Let $\scriptF_{sp}^{rel}(\rho_{sn},N_d)$ denote the formation of $N_d$ defenders uniformly distributed on the sphere of radius $\rho_{sn}$ centered at the origin and characterized by $\bm{\xi}_{l}^{sp_0}$, for all $l\in I_d$ (see for example $\scriptF_{sp}^{rel}(60,20)$ shown in Fig.~\ref{fig:three_formations}).
	
	We choose $\rho_{sn}$ such that even if all the triangular net-like barriers have sides with length $\bar{R}_{sb}$, the volume enclosed by the formation $\scriptF_{sp}^{rel}(\rho_{sn},N_d)$ contains a sphere of radius $\rho_{ac}$. This requires $\rho_{sn}\ge \sqrt{\rho_{ac}^2+\frac{(\bar{R}_{sb})^2}{3}}$. Additionally, we require $\rho_{sn}\ge \rho_{ac}+b_d$ where $b_d$ is the tracking error \cite{chipade2020swarmherding}. Due to practical limit of $\bar{R}_{sb}$ on the edge length, to obtain a formation with minimal number of defenders, $\rho_{sn}$ should be equal to its minimal value so we choose $\rho_{sn}=\underline{\rho}_{sn}=\max\{\sqrt{\rho_{ac}^2+\frac{(\bar{R}_{sb})^2}{3}}, \rho_{ac}+b_d\}$. 
	
	Given the radius of formation $\rho_{sn}=\underline{\rho}_{sn}$, we want to find the minimum number of defenders on the formation $\scriptF_{sp}^{rel}(\rho_{sn},N_d)$ that satisfy the practical constraints on the maximum edge length on the underlying closed-3D-StringNet (Condition \ref{cond:3D-StringNet}). This requires maximum edge length $R_{sb}^{max}=\max_{(j, k)\in \calE_{sn}^{cl}} \norm{\bm{\xi}_{j}^{s_0}-\bm{\xi}_{k}^{s_0}}$ on $\scriptF_{sp}^{rel}(\rho_{sn},N_d)$ be smaller than $\bar{R}_{sb}$.
	In Fig.~\ref{fig:f_N}, the black curve shows the values of $R_{sb}^{max}$ for different values of $N_d$ by numerically evaluating uniform formations $\scriptF_{sp}^{rel}(\rho_{sn},N_d)$ for the given values of $N_d$. As observed, finding an explicit function that maps $N_d$ to $R_{sb}^{max}$ on $\scriptF_{sp}^{rel}(\rho_{sn},N_d)$ is extremely difficult. The reason is that unlike circular formation, the symmetry is relatively rare in three-dimensional spherical formation. To remedy this, we compute the minimum $N_d$ by numerically enumerating the formations by using the steps Algorithm \ref{alg:min_number_defenders}. 
	\vspace{-3.5mm}
	\begin{algorithm}
		\lstset{numbers=left, numberstyle=\tiny, stepnumber=1, numbersep=5pt}
		\caption{Minimum number of defenders $N_d$}
		\label{alg:min_number_defenders}
		Initialize $N_d=N_{d0}$\\
		Find the distribution $\scriptF_{sp}^{rel}(\rho_{sn},N_{d})$ and $R_{sb}^{max}$\\
		\If{$R_{sb}^{max}$ does not satisfy Condition \ref{cond:3D-StringNet}}{Set $N_{d} =N_{d}+1$ and repeat step 2 to 3}
		\Return{$N_d$}
	\end{algorithm}
	\vspace{-4mm}
	Given the uncertain dependence of maximum edge length on $N_d$, one may be tempted to use minimum choice of $N_{d0}=4$ as an initial guess. However, this may require longer time to determine the best $N_d$ for larger $\rho_{sn}$. In Fig.~\ref{fig:f_N}, the red curve shows the average edge length $R_{sb}^{av}$ on $\scriptF_{sp}^{rel}$. We notice that the average length of edges $R_{sb}^{av}$ can be well fitted by a function $f_N(N_d)$:
	\be \label{f_n}
	\abovedisplayskip=3pt
	\belowdisplayskip=3pt
	\baa
	f_N(N_d)=\sqrt{\frac{2(1-2\cos(\frac{\pi N_d}{3N_d-6}))}{(1-\cos(\frac{\pi N_d}{3N_d-6}))}},
	\eaa
	\ee
	shown as the blue curve in Fig.~\ref{fig:f_N}. We have that the maximum length $R_{sb}^{max}$ satisfies: $f_N(N_d)<R_{rel}={R_{sb}^{max}}{\rho_{sn}}$. So we can safely choose
	$N_{d0}=f_N^{-1}(\frac{\bar{R}_{sb}}{\rho_{sn}})$ as the initial guess to the iterative scheme mentioned earlier to find minimum $N_d$ satisfying Condition \ref{cond:3D-StringNet}. 
	By doing so, we start closer to the desired minimum value of $N_d$ and the computational time to find this $N_d$ can be greatly reduced, as shown in Fig.~\ref{fig:N_m}, where $\Delta N$ represents number iterations required to find minimum $N_d$. 
	\begin{figure}[!h]
		\centering
		\includegraphics[width=0.5\textwidth,trim={1.3cm .1cm 1.4cm .75cm},clip]{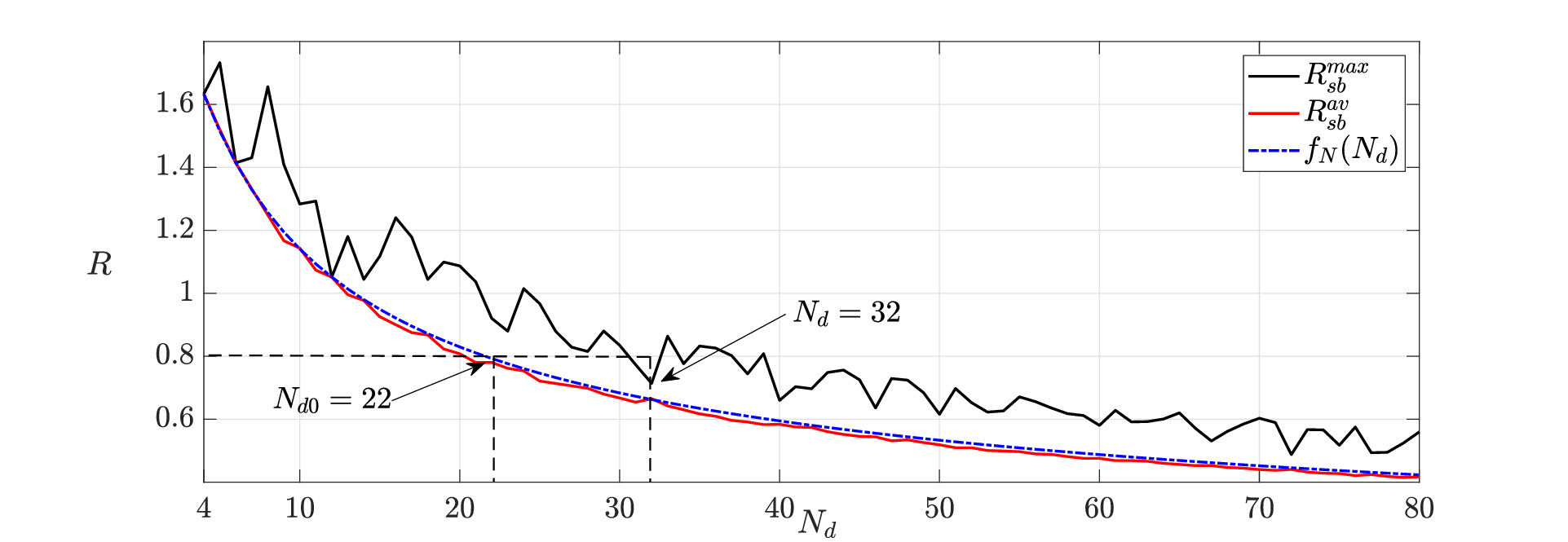}
		\caption{Relative edge lengths in the spherical formation}
		\vspace{-4mm}
		\label{fig:f_N}
	\end{figure}
	\begin{figure}[!h]
		\belowdisplayskip -5pt
		\centering
		\includegraphics[width=0.5\textwidth]{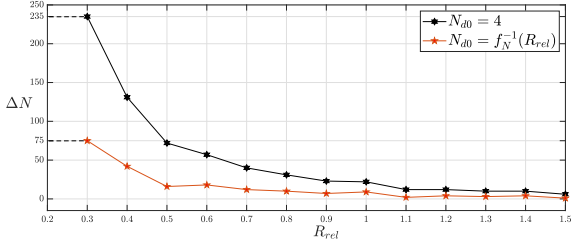}
		\caption{Number of iterations comparison}
		\vspace{-5mm}
		\label{fig:N_m}
	\end{figure}
	
	In practice, since the number of defenders is finite, the data of $R_{rel}=\frac{R_{sb}^{max}}{\rho_{sn}}$ for different values of $N_d$ can be pre-calculated. Then the problem of finding the minimum $N_d$ simply reduces to a linear search over the stored information which can be significantly faster compared to Algorithm~\ref{alg:min_number_defenders}. 
	
	\subsection{Intermediate 3D-StringNet Formations}
	Following the similar idea as in our 2D herding approach \cite{chipade2020swarmherding}, the defenders enclose the attackers via the closed-3D-StringNet, which is realized through a sequence of intermediate 3D-StringNet formations. We design two open-3D-StringNet formations for this purpose: 1) open-3D-StringNet $\calG_{sn}^{op}$ with hemispherical formation $\scriptF_{hs}^{rel}$, and 2) open-3D-StringNet $\calG_{sn}^{op}$ with planar formation $\scriptF_{pl}^{rel}$. 
	These formations are obtained by transforming the uniform spherical formation $\scriptF_{sp}^{rel}$ by using mappings that satisfy the Condition \ref{cond:3D-StringNet}. These mappings are discussed in the following subsections.
	
	
	\subsubsection{Mapping between hemispherical and spherical formation}
	Let $\mathbf{r}_{l}^{sp}=[\rho_{sn},\theta_l^{sp}, \phi_l^{sp}]^T=[\rho_{sn},\mathbf{p}_{l}^*]^T \in \bS_s\triangleq [0,\infty)\times[0,\pi]\times[-\pi,\pi]$ denote the $l^{th}$ desired position in $\scriptF_{sp}^{rel}$ in the spherical coordinates and $\mathbf{r}_{l}^{hs}=[\rho_{sn},\theta_l^{hs}, \phi_l^{hs}]^T \in \bS_h\triangleq [0,\infty)\times[0,\pi]\times[-\frac{\pi}{2},\frac{\pi}{2}]$ denote the $l^{th}$ desired position in $\scriptF_{hs}^{rel}$ in the spherical coordinates.
	We consider the mapping $\mathbf{m}_{sp}^{hs}:\bS_{sp}\rightarrow\bS_{hs}$ given by:
	\be \label{mapping1}
	\abovedisplayskip=3pt
	\belowdisplayskip=3pt
	\mathbf{r}_{l}^{hs}=\mathbf{m}_{sp}^{hs}(\mathbf{r}_{l}^{sp})=\left[\rho_{sn}, \; \theta_l^{sp}, \; 0.5\phi_l^{sp}\right]^T.
	\ee
	By mapping $\mathbf{m}_{sp}^{hs}$, the spherical formation is cut by the half plane $\phi=\pm\pi$ and then two sides of the cut rotate towards the plane $\phi=\pm\frac{\pi}{2}$ yielding a hemispherical shell like formation (Fig.~\ref{fig:three_formations}).
	
	We claim that all the edges in $\calG_{sn}^{op}$ on the hemispherical formation $\scriptF_{hs}^{rel}$ obtained through the mapping $\mathbf{m}_{sp}^{hs}$ satisfy the Condition \ref{cond:3D-StringNet}. To see why, consider the length of the edge $(\mathbf{r}_i^{hs},\mathbf{r}_j^{hs}) \in \calE_{sn}^{op}$: $
	L_{ij}^{hs}=\rho_{sn}\sqrt{2-2\Lambda(\Delta\phi_{ij}^{hs},\theta_i^{hs},\theta_j^{hs})}
	$. Similarly, the length of the edge $(\mathbf{r}_i^{sp},\mathbf{r}_j^{sp}) \in \calE_{sn}^{cl}$: $
	L_{ij}^{sp}=\rho_{sn}\sqrt{2-2\Lambda(\Delta\phi_{ij}^{sp},\theta_i^{sp},\theta_j^{sp})}
	$.  
	The only difference between $L_{ij}^{sp}$ and $L_{ij}^{hs}$ is that $\Delta\phi_{ij}^{hs}=\frac{1}{2}\Delta\phi_{ij}^{sp}$ and it is easy to see that $L_{ij}^{hs}\leq L_{ij}^{sp}$. These desired positions $\mathbf{r}_{l}^{hs}$ are represented in Cartesian coordinates by $\bm{\xi}_{l}^{hs_0} =\rho_{sn}[\sin(\theta_l^{hs})\cos(\phi_l^{hs}), \; \sin(\theta_l^{hs})\sin(\phi_l^{hs}), \; \cos(\theta_l^{hs})]^T\in \bR^3$, for all $l \in I_d$.

	\subsubsection{Mapping between planar and hemispherical formation}
	
	For a given constraint on the edge length, a planar formation will create a larger blockage in the path of the attackers as compared to the hemispherical one. Therefore, an open-3D-StringNet $\calG_{sn}^{op}$ with planar formation $\scriptF_{pl}^{rel}$ is chosen as the desired formation to be achieved at the end of the gathering phase.
	
	The planar formation $\scriptF_{pl}^{rel}$ is obtained from $\scriptF_{hs}^{rel}$. To ease out the mathematics, $\scriptF_{hs}^{rel}$ is first rotated about the cartesian $y$-axis by $90^\circ$ to obtain a rotated formation $\scriptF_{hs'}^{rel}$ (Fig.~\ref{fig:three_formations}). Let   
	$\mathbf{r}_{l}^{hs'} =[\rho_{sn},\theta_l^{hs'},\phi_l^{hs'}]^T \in \bS_{hs'}=[0,\infty)\times[0,\frac{\pi}{2}]\times[0,2\pi)$ be the position corresponding to $\mathbf{r}_l^h$ after the aforementioned rotation. Let $\mathbf{r}_{l}^{pl}=[\rho_l^{pl},\phi_l^{pl}]^T \in \bC_{pl}\triangleq [0,\infty)\times[0,2\pi)$ be the $l^{th}$ desired position in the planar formation $\scriptF_{pl}^{rel}$. We consider a mapping $\mathbf{m}_{hs'}^{pl}:\bS_{hs'}\rightarrow\bC_{pl}$ given by
	\be \label{mapping2}
	\abovedisplayskip=3pt
	\mathbf{r}_{l}^{pl}=\mathbf{m}_{hs'}^{pl}(\mathbf{r}_{l}^{hs'})=[k_{pl}\rho_{sn}\sin(\theta_j^{hs'}), \; \phi_j^{hs'}]^T,
	\ee
	where $k_{pl}$ is a constant scaling factor.
	The lengths of the edges in $\calG_{sn}^{op}$ corresponding to the formations $\scriptF_{hs'}^{rel}$ and $\scriptF_{pl}^{rel}$ denoted as $L_{ij}^{hs'}$ and $L_{ij}^{pl}$, respectively, are given by
	\be \label{len_p}
	\arraycolsep=1pt
	\abovedisplayskip=0pt
	\baa
	L_{ij}^{hs'}&=\rho_{sn}\sqrt{2-2\Lambda(\Delta\phi_{ij}^{hs'},\theta_i^{hs'},\theta_j^{hs'})},\\
	L_{ij}^{pl}&=\sqrt{(\rho_i^{pl})^2+(\rho_j^{pl})^2-2(\rho_i^{pl})(\rho_j^{pl})C(\Delta\phi_{ij})}\\
	&\leq k_{pl}\rho_{sn}\sqrt{2-2\Lambda(\Delta\phi_{ij}^{hs'},\theta_i^{hs'},\theta_j^{hs'})}=k_{pl}L_{ij}^{hs'}.
	\eaa
	\ee
	We have the following result.
	
	\begin{lemma}\label{lem:k_pl}
		If $0<k_{pl}\le \frac{\bar{R}_{sb}}{R_{hs}^{max}}$, then $\calG_{sn}^{op}$ with planar formation $\scriptF_{pl}^{rel}$ satisfies Condition~\ref{cond:3D-StringNet}, where $R_{hs}^{max}=\max_{(j, k)\in \calE_{sn}^{op}} \norm{\mathbf{r}_{j}^{hs'}-\mathbf{r}_{k}^{hs'}}$ is the length of the longest edge on the hemispherical formation. Furthermore, we have $\frac{\bar{R}_{sb}}{R_{hs}^{max}}=\frac{\bar{R}_{sb}}{R_{sb}^{max}}\frac{R_{sb}^{max}}{R_{hs}^{max}}>1$.
	\end{lemma}
	Lemma~\ref{lem:k_pl} implies that, by choosing $k_{pl}>1$, the mapping $\mathbf{m}_{hs'}^{pl}$ is able to generate a circular planar formation $\scriptF_{pl}^{rel}$ with radius $\rho_{sn,pl}>\rho_{sn}$ that satisfies Condition \ref{cond:3D-StringNet}. These desired positions $\mathbf{r}_{l}^{pl}$ are represented in Cartesian coordinate system by $\bm{\xi}_{l}^{pl_0}=\rho_l^{pl}[\cos(\phi_l^{pl}),\sin(\phi_l^{pl}),0]^T \in \bR^3$, for all $l \in I_d$. We call the local body-fixed $z$-axis as the orientation vector of the formation $\scriptF_{pl}^{rel}$.

	\section{Modifications to 2D StringNet Herding}\label{sec:herding}
	The defenders follow the same overall structure of the 2D-StringNet herding \cite{chipade2020swarmherding}, while utilizing the 3D-StringNet formations generated in the previous section and with appropriate modifications to the corresponding parts from the 2D approach. Thus, the 3D StringNet herding consists of four phases \cite{chipade2020swarmherding}: 1) Gathering and forming a planar formation. 2) Seeking the attackers while maintaining the planar formation. 3) Enclosing the attackers by forming a spherical formation around them. 4) Herding the enclosed attackers to $\calS$. These phases are discussed in the following subsections.

	\subsubsection{Gathering}
	In the gathering phase, the defenders first converge to the planar formation $\scriptF_{pl}^g$ centered at the gathering center $\mathbf{r}_{df^g}$ on the expected path of the attackers (shortest path to the protected area). Let us define a mathematical object $\calT \calR$ to define formations obtained by translating and rotating a given formation $\scriptF$. We obtain $\scriptF_{pl}^g$ by translating the formation $\scriptF_{pl}^{rel}$ to $\mathbf{r}_{df^g}$ and rotating by $\calR(\mathbf{q}_{ac})$, where $\calR(\mathbf{q}_{ac})$ is the rotation matrix corresponding to the orientation represented by the quaternion $\mathbf{q}_{ac}$, where $\mathbf{q}_{ac}$ denotes the orientation when body $z$-axis points toward the attackers' center $\mathbf{r}_{ac}$. We denote this transformation by $\scriptF_{pl}^g = \calT \calR (\mathbf{r}_{df^g},\mathbf{q}_{ac}) \scriptF_{pl}^{rel}$. In particular, the formation $\scriptF_{pl}^g$, with underlying graph $\calG_{sn}^{op}$, is characterized by positions $\bm{\xi}_{\text{a}(j)}^g=\mathbf{r}_{df^g}+\calR(\mathbf{q}_{ac}) \bm{\xi}_{\text{a}(j)}^{p_0}$ for all $j \in I_d$. The gathering center $\mathbf{r}_{df^g}$ of the gathering formation $\scriptF_{pl}^g$ is obtained by solving a mixed integer quadratic program (MIQP) iteratively \cite{chipade2020swarmherding}. The defender $\calD_j$ converges to its assigned desired (goal) position $\bm{\xi}_{\text{a}(j)}^g$ on $\scriptF_{pl}^g$, where $\text{a}:I_d \rightarrow I_d$ is the defender-goal assignment obtained from the MIQP \cite{chipade2020swarmherding}. After the defenders arrive at their desired positions, they establish nets with the neighboring defenders as per $\calF_{sn}^{op}$. Then, the defending swarm enters the seeking phase which is discussed next.
	
	\subsubsection{Seeking}
	In practice, the attackers may deviate from their optimal trajectories computed during the gathering phase, which requires defenders to come closer to the attackers in order to enclose them.
	In the seeking phase, we consider the desired formation $\scriptF_{pl}^s=\calT \calR(\mathbf{r}_{df^s},\mathbf{q}_{df^s})\scriptF_{pl}^{rel}$ of the defenders as a virtual rigid body with center of mass at $\mathbf{r}_{df^s}$, where $\mathbf{q}_{df^s}=[q_1,q_2,q_3,q_4]^T=[\tilde{\mathbf{q}}_{df^s}^T,q_4]^T$ is the quaternion that represents the orientation of the formation $\scriptF_{pl}^s$. The virtual body's translational motion is governed by the same dynamics as in \eqref{eq:defendDyn1} and the rotational dynamics are governed by Euler equations and quaternion kinematics:
	\be \label{eq: q_update1}	
	\abovedisplayskip=0pt
	\belowdisplayskip=1pt
	\baa
	\dot{\tilde{\mathbf{q}}}_{df^s}&\hspace{-3mm}=0.5(\bm{\omega}_{df^s} \tilde{\mathbf{q}}_{df^s}+q_4\bm{\omega}_{df^s}),\quad \dot q_4=0.5\bm{\omega}_{df^s}^T\tilde{\mathbf{q}}_{df^s};
	\\
	\dot{\bm{\omega}}_{df^s}&\hspace{-3mm}=\mathbf{u}_{df^s}^{rot}, 
	\eaa
	\ee
	where $\bm{\omega}_{df^s}=[\omega_x,\omega_y,\omega_z]^T$ is the angular velocity of the rigid body resolved in body-fixed frame.
	To ensure that the desired formation gets closer to the attackers and the orientation of the formation faces the attackers, we apply the following translational and rotational feedback accelerations to the virtual rigid body \cite{wie1989quaternion}:
	\bse
	\abovedisplayskip=3pt
	\belowdisplayskip=3pt
	\begin{align}\label{eq:def_des_control1}
		&\mathbf{u}_{df^s}^{trans}=\bm{\Omega}_{\bar{u}_{df^s}^{trans}} \left( - k_1(\mathbf{r}_{df^s}-\mathbf{r}_{ac})\right),\\
		&\mathbf{u}_{df^s}^{rot}=\bm{\Omega}_{\bar{u}_{df^s}^{rot}}  \left (-D\bm{\omega}_{df^s}-K\mathbf{q}_e \right)
	\end{align} 
	\ese
	where $\bar{u}_{df^s}^{trans}$ and $\bar{u}_{df^s}^{rot}$ are saturation limits; $k_1$, $K$ and $D$ are gain matrix which are diagonal matrices with non-negative scalars \cite{wie1989quaternion}. The quaternion $\mathbf{q}_{des}$ represents the desired orientation where the local $z$-axis points toward the center of attackers $\mathbf{r}_{ac}$. $\mathbf{q}_e=Q(\mathbf{q}_{des})\mathbf{q}_{df^s}$ is the attitude error between the current quaternion and $\mathbf{q}_{des}$.  
	The initial quaternion is $\mathbf{q}_{df^s}(0)=\mathbf{q}_{ac}$ and the initial angular velocity is $\bm{\omega}=[0,0,0]^T$.
	
	The desired position $\bm{\xi}_{l}^s=\mathbf{r}_{df^s} +\calR(\mathbf{q}_{df^s}) \bm{\xi}_{l}^{pl_0}$, for $l \in I_d$, on the desired formation $\scriptF_{pl}^s$ satisfies:
	\bse \label{eq:desired_formation_seeking}
	\abovedisplayskip=1pt
	\belowdisplayskip=2pt
	\arraycolsep=1.4pt
	\begin{align*}
		\dot{\bm{\xi}}_{l}^s=&\bm{\eta}_{l}^s=\dot{\mathbf{r}}_{df^s}+\bm{\omega}_{df^s}\times\bm{\xi}_{l}^{pl_0},\\ 
		\dot{\bm{\eta}}_{l}^s=&\mathbf{u}_{df^s}^{trans}-C_d \norm{\mathbf{v}_{df^s}}\mathbf{v}_{df^s}+\dot{\bm{\omega}}_{df^s}\times\bm{\xi}_{l}^{pl_0}\\
		&+\bm{\omega}_{df^s}\times(\bm{\omega}_{df^s}\times\bm{\xi}_{l}^{pl_0}).
	\end{align*}
	\ese
	The defenders $\calD_j$ track their assigned desired position $\bm{\xi}_{\text{a}(j)}^s$ using the 3D extension of the 2D finite-time convergent controllers as in \cite{chipade2020swarmherding}. Seeking is completed when $\norm{\mathbf{r}_{df^s}-\mathbf{r}_{ac}}<\epsilon_1$ and $\mathbf{q}_e<\epsilon_2$, where $\epsilon_1>0$ and $\epsilon_2>0$ are user defined small thresholds.
	
	\subsubsection{Enclosing}
	After the defenders come close to the attackers as an open-3D-StringNet with $\scriptF_{pl}^s$ at the end of seeking, the enclosing phase is initiated. In the enclosing phase, defenders aim to enclose the attackers in the closed-3D-StringNet with formation $\scriptF_{sp}^e=\calT \calR (\mathbf{r}_{ac}, \mathbf{q}_{df^e})\scriptF_{sp}^{rel}$, where $\mathbf{q}_{df^e}$ is the quaternion at the end of the seeking phase. Starting from the planar formation $\mathcal{F}_p^s$, the defenders first achieve an open-StringNet with hemispherical formation $\scriptF_{hs}^e=\calT \calR (\mathbf{r}_{ac}, \mathbf{q}_{df^e})\scriptF_{hs}^{rel}$, and then the closed-3D-StringNet with formation $\scriptF_{sp}^e$. The reason to choose an intermediate open-3D-StringNet formation $\scriptF_{hs}^e$ is to avoid that the defenders unnecessarily come close to each other while converging to $\scriptF_{sp}^e$ allowing the attackers to disperse. 
	The control actions for the defenders to track their desired positions on the respective formations during this phase can be obtained from \cite{chipade2020swarmherding}. The desired formation $\scriptF_{hs}^e$ is switched to $\scriptF_{sp}^e$ when the defenders come within a distance of $b_d$ from their desired positions on $\scriptF_{hs}^e$. The closed-3D-StringNet is achieved when all defenders converge to their desired locations, i.e., $\norm{\mathbf{r}_{dj}-\bm{\xi}_{\text{a}(j)}^e}<b_d$ for all $j \in I_d$, where $b_d$ is the tracking error incurred due to the unknown but bounded acceleration terms $\ddot{\bm{\xi}}_{\text{a}(j)}$ \cite{chipade2020swarmherding}.
	
	\subsubsection{Herding}
	Once the defenders form the closed-3D-StringNet around the attackers, they move towards the safe area while tracking a rigid spherical formation $\mathscr{F}_{sp}^h=\calT \calR (\mathbf{r}_{df^h},\mathbf{q}_{df^h})\scriptF_{sp}^{rel}$ centered at a virtual agent $\mathbf{r}_{df^h}$, where $\mathbf{q}_{df^h}$ is equal to $\mathbf{q}_{df^e}$ at the start of the herding phase. The virtual agent moves towards the safe area $\calS$ as discussed in \cite{chipade2020swarmherding} and the defenders use the finite-time, bounded tracking controllers similar to that in \cite{chipade2020swarmherding} to track their desired positions on $\scriptF_{sp}^h$. The herding phase ends when every enclosed attacker is successfully herded into the safe area.

	\section{Dominance Region for the Defenders} \label{sec:dominance_region}
	
	The success of the defenders depends on whether they are able to achieve the open-3D-StringNet with planar formation $\scriptF_{pl}^g$ in the expected path of the attackers, well before the attackers reach the gathering center. For given initial conditions of all the agents, the defenders require to solve the problem of finding the best gathering center $\mathbf{r}_{df^g}$ and the corresponding defender-goal assignment $\text{a}$ using the iterative MIQP formulation \cite{chipade2020swarmherding}, which becomes computationally demanding as the number of agents becomes larger. In this section, we characterize the conditions on the initial positions of the attackers for which the defenders are able to achieve the formation $\scriptF_{pl}^g(\mathbf{r}_{df^g},\mathbf{q}_{ac})$ at a location $\mathbf{r}_{df^g}$ on the shortest path of the attackers to the protected area, before the attackers can reach there. We call this set of initial conditions of the attackers as the dominance region for the given initial positions of the defenders. Let $T_a(\mathbf{r}_{a},\mathbf{r},\rho_{a})$ be the minimum time required by an attacker at $\mathbf{r}_{a}$ to reach within $\rho_{a}$ distance from the point $\mathbf{r}$. Let $\mathbf{R}_d=[\mathbf{r}_{d1},\mathbf{r}_{d2},...,\mathbf{r}_{dN_d}]$ denote the positions of the defenders $\calD_j$ for all $j \in I_d$. Let $T_d(\mathbf{R}_{d},\scriptF_{pl}^g(\mathbf{r},\mathbf{q}))$ be the maximum time required by all the defenders to achieve the gathering formation $\scriptF_{pl}^g(\mathbf{r},\mathbf{q}))$ centered at $\mathbf{r}$. The dominance region is then formally defined as:
	\begin{definition}[Defenders' Dominance Region]
		$Dom(\mathbf{R}_{d},\bar{\rho}_{ac},\Delta T_d^g)=\{\mathbf{r}\in \bR^3| \exists \upsilon \in (\frac{\rho_p}{\norm{\mathbf{r}}},1-\frac{\bar{\rho}_{ac}}{\norm{\mathbf{r}}})$ such that $T_a(\mathbf{r},\mathbf{r}_{df^g},\bar{\rho}_{ac})-T_d(\mathbf{R}_{d},\scriptF_{pl}^g(\mathbf{r}_{df^g},\mathbf{q}_{ac}))))\ge \Delta T_d^g $ $\text{where }\mathbf{r}_{df^g}= \upsilon \mathbf{r} \}$, where $\Delta T_d^g$ is a user-defined time to account for the size of the attackers' swarm and the time required by the defenders to get connected by triangular net-like barriers once arrived at the desired formation.
	\end{definition}
	
	We provide the following formulation that is based on approximation functions, and is computationally less intensive, to find an estimate $Dom_{est}$ of the dominance region $Dom$ that is completely contained inside $Dom$.
	
	Consider $N_d$ defenders and $N_a$ attackers located at given positions as shown in Fig.\ref{fig:abstractionForDominanceRegion}. Let the largest radius of the attackers' formation be $\bar{\rho}_{ac}$. Consider the protected area located at the origin ($\mathbf{r}_p=[0,0,0]^T$). Let the center of mass of the attackers have spherical coordinates ($R_{ac},\phi_{ac}, \theta_{ac}$). Consider the gathering center $\mathbf{r}_{df^g}$ at $(R,\phi_{ac}, \theta_{ac})$. The distance of the defender $\calD_j$ from the center of the gathering formation (Fig.~\ref{fig:abstractionForDominanceRegion}) is:
	\be
	\varrho_j=\sqrt{R^2+R_j^2-2RR_j \Lambda(\phi_{ac}-\phi_{dj}, \theta_{ac},\theta_{dj})},
	\ee
	where $(R_j, \phi_{dj},\theta_{dj})$ are spherical coordinates of the defender $\calD_j$'s position for all $j \in I_d$.
	We have the following proposition using the approximation of maximum function as in \cite{stipanovic2012monotone}.
	\begin{proposition}
		The maximum value among $\varrho_j$, $j \in I_d$, satisfies: $\bar{\varrho}=\ds \max_{j \in I_d} \varrho_j \le \tilde{\varrho}_\delta$ $= \sqrt[\delta]{\sum_{j \in I_d} \varrho_j^{\delta}}$ and $\ds \lim_{\delta \rightarrow \infty} \tilde{\varrho}_\delta=\bar{\varrho}$. 
	\end{proposition} 
	\begin{figure}[h]
		\centering
		\setlength{\abovecaptionskip}{2pt plus 3pt minus 2pt}
		\setlength{\belowcaptionskip}{-10pt plus 3pt minus 2pt}
		\includegraphics[width=.9\linewidth,trim={3.1cm 2.5cm 1.1cm 1.45cm},clip]{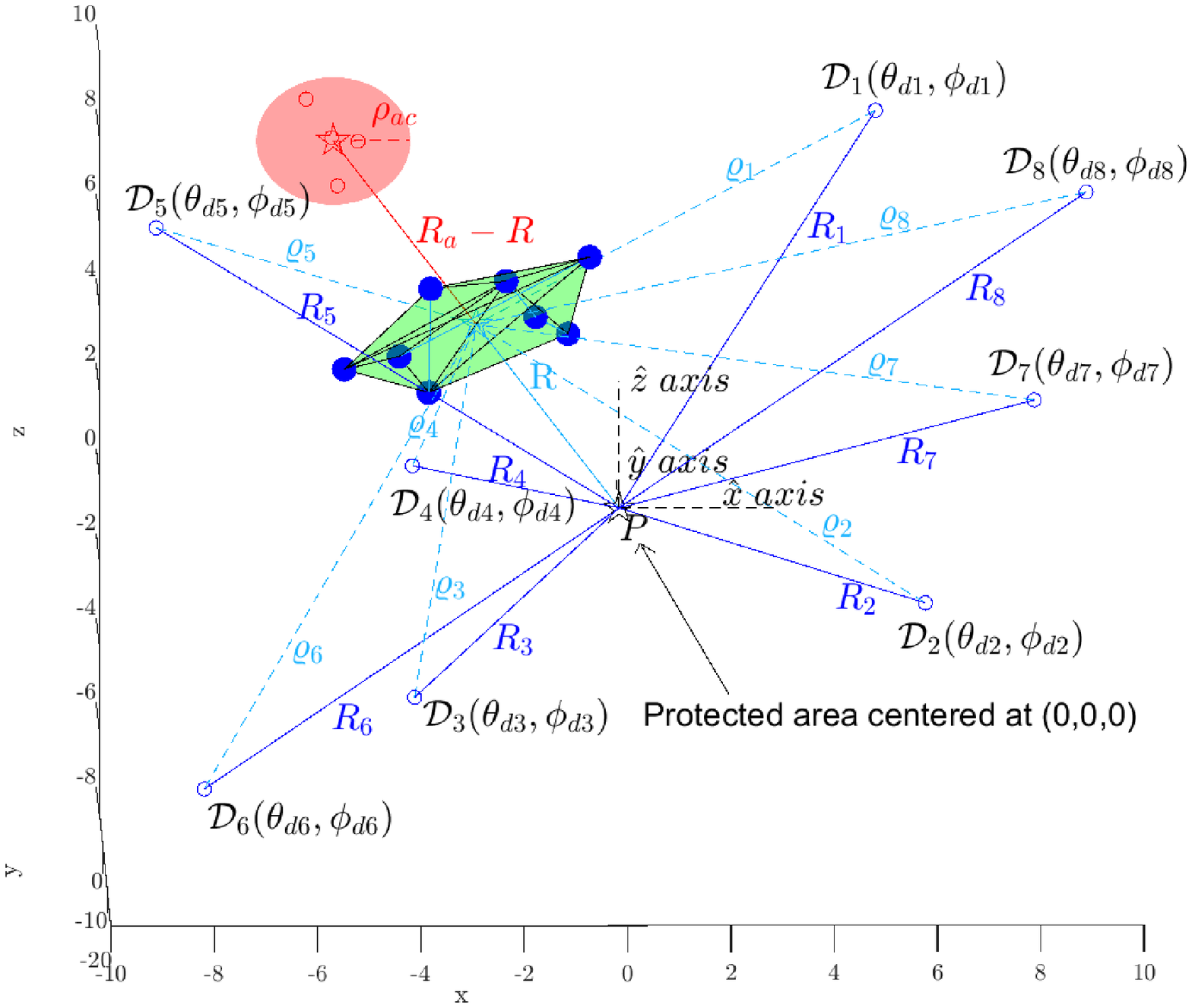}
		\caption{Abstraction for estimate of dominance region}
		\label{fig:abstractionForDominanceRegion}
	\end{figure}
	The maximum distance any defender would have to travel in the best defender-goal assignment can be upper bounded by $\bar{\varrho}_d=\tilde{\varrho}_{\delta}+ \rho_{sn,p}$, where $\rho_{sn,p}$ is the radius of the planar gathering formation $\scriptF_{pl}^g$. The maximum time for any defender to reach the gathering location assigned to it as per the best defender-goal assignment under time-optimal control \cite{chipade2020approximate} can be upper bounded by: 
	\be \label{eq:maxDefenderTime}
	\baa
	\bar{T}_d(\bar{\varrho}_d) = \frac{1}{\lambda_0} \big(\tanh^{-1} \big(\frac{v_{sw}}{\bar{v}_d} \big) + \tan^{-1} \big(\frac{v_{sw}}{\bar{v}_d} \big)
	\big),
	\eaa
	\ee
	where $\lambda_0=\sqrt{\bar{u}_dC_D}$, $v_{sw}=\sqrt{\frac{(\lambda-1)\bar{u}_d}{(\lambda+1)C_D}}$, $\lambda=e^{2C_D\bar{\varrho}_d}$. Similarly, the minimum time that the attackers require to reach the gathering location is when the attackers move towards the protected with the maximum possible speed.
	The difference between the time needed by the attackers to reach the gathering center and the time required by the defenders to reach there can be bounded from below by:
	\be
	\baa
	\Delta T =  \frac{R_{ac}-\bar{\rho}_{ac}-R}{\bar{v}_a} - \bar{T}_d(R)
	\eaa
	\ee
	Defenders want $\Delta T\ge \Delta T_d^g$ to be able to gather well before the attackers reach the gathering center. We are interested in the limiting condition when $\Delta T=\Delta T_d^g$, for which we have:
	\be \label{eq:AttackDefendTimeDiff1}
	R_{ac}= f(R) =\bar{\rho}_{ac}+R+\bar{v}_a(\bar{T}_d(R)+\Delta T_d^g).
	\ee
	We want to find the smallest value $\underline{R}_{ac} (> \rho_p)$ of $R_{ac}$ for which $\Delta T =\Delta T_d^g$, i.e.,
	\be \label{prob:find_R}
	\underline{R}_{ac} =\textstyle \min_{R>\rho_p} f(R).
	\ee
	
	%
	\begin{lemma}
		Given that no two defenders are co-located, i.e., $\norm{\mathbf{r}_{dj}-\mathbf{r}_{dj'}}>0$ for all $j \neq j' \in I_d$, $f(R)$ as given in Eq.~\eqref{eq:AttackDefendTimeDiff1} is a locally convex function of $R$.
	\end{lemma}
	\begin{proof}
		Sum of two convex functions is always a convex function \cite{boyd2004convex}, so it is sufficient to show that $\bar{T}_d(R)$ is a locally convex function to show that $f(R)$ is a locally convex function. Let $g(R)=\bar{T}_d(\bar{\varrho}_d(R))$. The double derivative of $g$ is:
		\be\label{eq:partial_g}
		\baa
		\doublepdydx{g}{R} & = \doublepdydx{\bar{T}_d}{\bar{\varrho}_d}\left(\pdydx{\bar{\varrho}_d}{R}\right)^2 + \pdydx{\bar{T}_d}{\bar{\varrho}_d}\doublepdydx{\bar{\varrho}_d}{R}
		.
		\eaa
		\ee
		We have 
		\bse 
		\begin{align}
			\textstyle 
			\pdydx{\bar{T}_d}{\bar{\varrho}_d}=& \textstyle \frac{C_d}{\lambda_0}\sqrt{\frac{\lambda+1}{\lambda-1}} \ge 0;\label{eq:partial_Tbar}\\
			\textstyle
			\doublepdydx{\bar{T}_d}{\bar{\varrho}_d} = & \textstyle 
			\frac{1}{\lambda_0} \left( \frac{(2C_d)^2\lambda}{1-\lambda^2} \sqrt{\frac{\lambda+1}{\lambda-1}} \right) \le 0; \label{eq:double_partial_Tbar}\\
			\textstyle  \pdydx{\bar{\varrho}_d}{R}=& \textstyle \sum_{j=1}^{N_d} \varrho_{j}^{\delta-2} \left (\tilde{\varrho}_{\delta}\right)^{\frac{1}{\delta}-1} (R-R_{\Lambda j});\\
			\textstyle  \doublepdydx{\bar{\varrho}_d}{R}=& \textstyle  \sum_{j=1}^{N_d}  \left (\tilde{\varrho}_{\delta}\right)^{\frac{1}{\delta}-1}\varrho_{j}^{\delta-2}  \big \{ (\frac{1}{\delta}-1)  \frac{(R-R_{\Lambda j})}{(\tilde{\varrho}_{\delta}) } \pdydx{\bar{\varrho}_d}{R} \nonumber\\ 
			& \textstyle  1 + (\delta-2)\varrho_{j}^{-2}  {(R-R_{\Lambda j})^2}  \big \}, \label{eq:double_partial2}
		\end{align}
		\ese
		where $R_{\Lambda j}=R_j \Lambda(\phi_{ac}-\phi_{dj}, \theta_{ac},\theta_{dj})$.
		Let $R^*$ be such that  $\pdydx{\bar{\varrho}_d}{R}\vert_{R=R^*}=0$. We have that $\varrho_j$ is a convex function of $R$ which implies that its $\ell_\delta$-norm, $\tilde{\varrho}_{\delta}$, is also a convex function \cite{boyd2004convex}. This means $\tilde{\varrho}_{\delta}(R^*)$ is the minimum value of $\tilde{\varrho}_{\delta}$, i.e., $\tilde{\varrho}_{\delta}\ge \tilde{\varrho}_{\delta}(R^*)$. Since not all defenders are co-located $\tilde{\varrho}_{\delta}(R^*)>0$ implying $\tilde{\varrho}_{\delta}>0$ and $\lambda >1$. From Eq.~\eqref{eq:double_partial2}, we have $\doublepdydx{\bar{\varrho}_d}{R}\vert_{R=R^*} > 0$. Then from Eq.~\eqref{eq:partial_g}, we get $\doublepdydx{g}{R}\vert_{R=R^*}>0$. We know that $\varrho_j$ is a twice continuously differentiable function of $R$ for $R>0$ and if we choose $\delta \ge 2$ then we can show that both $\pdydx{\bar{\varrho}_d}{R}$ and $\doublepdydx{\bar{\varrho}_d}{R}$ are continuous functions of $R$. From Eq.~\eqref{eq:partial_Tbar} and~\eqref{eq:double_partial_Tbar}, we have that $\pdydx{\bar{T}_d}{\bar{\varrho}_d}$ and $\doublepdydx{\bar{T}_d}{\bar{\varrho}_d}$ are continuous functions of $R$. This implies that $\doublepdydx{g}{R}$ is continuous at $R=R^*$.
		
		Combining the two results that $\doublepdydx{g}{R}$ is continuous  and greater than 0 at $R=R^*$ implies that there exists $\epsilon>0$ such that $\doublepdydx{g}{R}>0$ for all $R$ satisfying $|R-R^*|<\epsilon$, i.e., $g(R)$ is locally convex in the neighborhood of $R=R^*$ and so is $f(R)$.
	\end{proof}
	One can find $\underline{R}_{ac}$ by solving the convex optimization \eqref{prob:find_R} with $R=R^*$, the minimizer of $\tilde{\varrho}_{\delta}(R)$, as an initial guess to a gradient descent algorithm with sufficiently small step size.
	
	Given the direction from which the attackers are approaching the protected area, one can solve the problem in \eqref{prob:find_R} to assess, at least conservatively, whether the defenders can gather in the attackers' path before the attackers, without solving the actual, computationally heavy iterative MIQP formulation \cite{chipade2020swarmherding}. Figure \ref{fig:DominanceRegions} shows the boundaries $\partial Dom_{est}$ and $\partial Dom$ of the estimate $Dom_{est}$ and the dominance region $Dom$, respectively. Here $\partial Dom_{est}$ is obtained by solving a simple quadratic program \eqref{prob:find_R} while $\partial Dom$ is obtained by numerically evaluating the iterative MIQP for each direction. The regions outside of the closed boundaries $\partial Dom_{est}$ and $\partial Dom$ are, respectively, $Dom_{est}$ and $Dom$, computed for the case where the defenders are at given locations (blue circles). On the other hand, the set inside the boundaries $\partial Dom_{est}$ and $\partial Dom$ are the complement sets $Dom_{est}^c=\bR^3 \backslash Dom_{est}$ and $Dom^c=\bR^3 \backslash Dom$, respectively. The set $Dom^c$ is essentially the dominance region of the attackers, i.e., the attackers can reach the protected area before the defenders can gather on their path if the attackers start inside $Dom^c$. Note that the estimate $Dom_{est}$ is completely contained in the dominance region $Dom$. The region $Dom$ is larger on the side where the density of the defenders is larger. This is intuitive because many defenders have to travel less when the attackers approach from this side and hence allow defenders to gather on the expected path of the attackers in time even if the attackers start more closer to the protected area on this side. We have the following result.
	\begin{figure*}[h]
		\centering
		\setlength{\abovecaptionskip}{2pt plus 3pt minus 2pt}
		\setlength{\belowcaptionskip}{-10pt plus 3pt minus 2pt}
		\includegraphics[width=.7\linewidth,trim={4cm 3.7cm 7.7cm 3.8cm},clip]{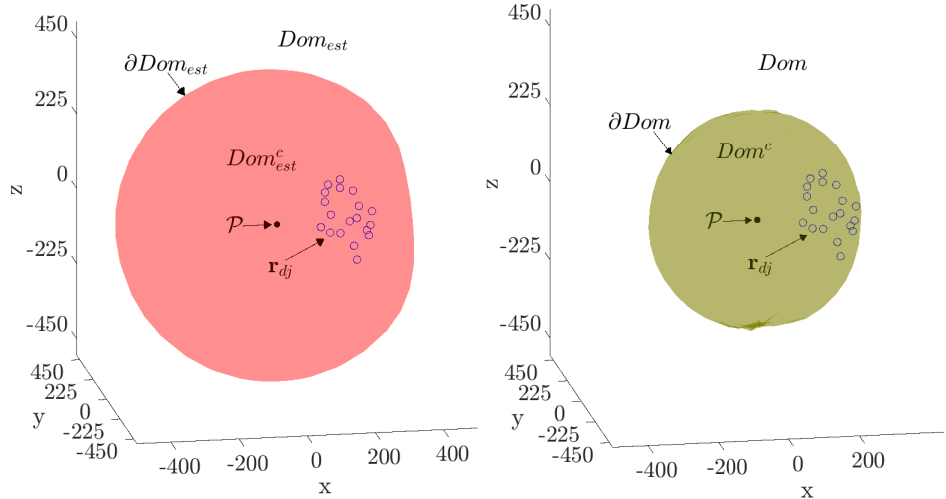}
		\caption{Dominance regions of the players (right: actual dominance region, left: estimate of the dominance region)}
		\label{fig:DominanceRegions}
	\end{figure*}
	
	\begin{theorem}\label{thm:dominance_regions}
		Consider a group of defenders $\calD=\{\calD_1,\calD_2,...\calD_{N_{dc}}\}$ starting at given locations $\mathbf{R}_{d}=[\mathbf{r}_{d1},\mathbf{r}_{d2},...,\mathbf{r}_{dN_d}]$ and a swarm of Attackers $\calA$ with maximum connectivity radius $\bar{\rho}_{ac}$. The defenders in $\calD$ are guaranteed to achieve a planar formation $\scriptF_{pl}^g$, located at a position on the shortest path from the center of mass of the attackers in $\calA$ to the protected area $\calP$, $\Delta T_d^g \;s$ before the attackers reach that position, if the attackers start inside $Dom_{est}(\mathbf{R}_{d},\bar{\rho}_{ac},\Delta T_d^g)$
	\end{theorem}
	\begin{proof}
		By construction, $Dom_{est}(\mathbf{R}_{d},\bar{\rho}_{ac},\Delta T_d^g) \subseteq Dom(\mathbf{R}_{d},\bar{\rho}_{ac},\Delta T_d^g)$. The proof follows from the definition of the dominance region $Dom(\mathbf{R}_{d},\bar{\rho}_{ac},\Delta T_d^g)$.
	\end{proof}
	In other words, Theorem~\ref{thm:dominance_regions} states that for the attackers starting in $Dom_{est}(\mathbf{R}_{d},\bar{\rho}_{ac},\Delta T_d^g)$ the defenders are guaranteed to gather in their shortest path to the protected area in time. However, if the attackers do not start in $Dom_{est}(\mathbf{R}_{d},\bar{\rho}_{ac},\Delta T_d^g)$ nothing can be concretely said about the gathering of the defenders based on the above approximate analysis.
	
	\section{Simulations}\label{sec:simulations} 
		In this section, 20 defending agents are deployed in a three-dimensional obstacle-free environment and they aim to protect the area $\calP$ by herding an adversarial swarm of 6 attackers to $\calS$. $\mathcal{B}_{\rho_{ac}}(\mathbf{r}_{ac})$ represents the connectivity region of attackers with radius $\rho_{ac}$. Fig.~\ref{fig:3DSwarmHerd1} shows that a circular planar formation is formed at the desired position facing towards the adversarial swarm. As observed in Fig.~\ref{fig:3DSwarmHerd2}, the planar formation gradually transforms into the hemispherical StringNet while tuning its attitude so that the hemispherical formation can be formed in a good position. After the hemispherical formation is constructed, the closed-3D-StringNet formation is quickly established and thus all of the attackers are contained, as shown in Fig.~\ref{fig:3DSwarmHerd3}. In Fig.~\ref{fig:3DSwarmHerd4}, the closed-3D-StringNet herds all the enclosed attackers directly towards the safe area. All the enclosed attackers are taken inside the safe area and the herding is completed. Video of the simulation can be found at \href{https://drive.google.com/drive/folders/1zxMXK5tpotrSKiTkH5yUyKU7dPL1ilnF?usp=sharing}{https://tinyurl.com/yyoonbd8}
	
	\begin{figure*}[t]
		\centering
		\begin{subfigure}[h]{0.48\linewidth}
			\centering
			\includegraphics[width=1\linewidth,trim={.1cm 2.3cm 2.cm 3.3cm},clip]{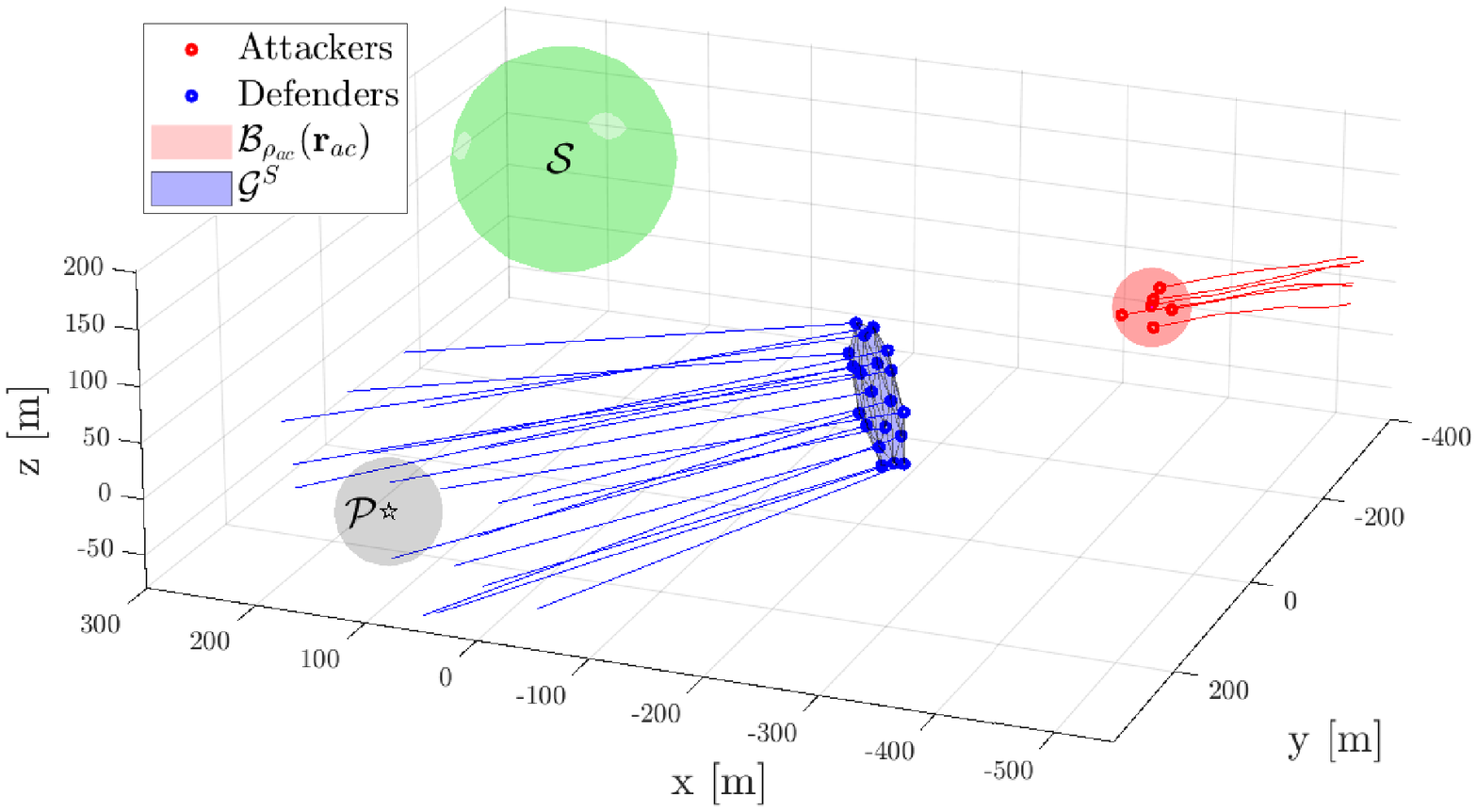}
			\caption{Gathering phase: Planar StringNet}
			\label{fig:3DSwarmHerd1}
		\end{subfigure}	
		\begin{subfigure}[h]{0.48\textwidth}
			\centering
			\includegraphics[width=1\linewidth,trim={.9cm 2.3cm 2.cm 3.4cm},clip]{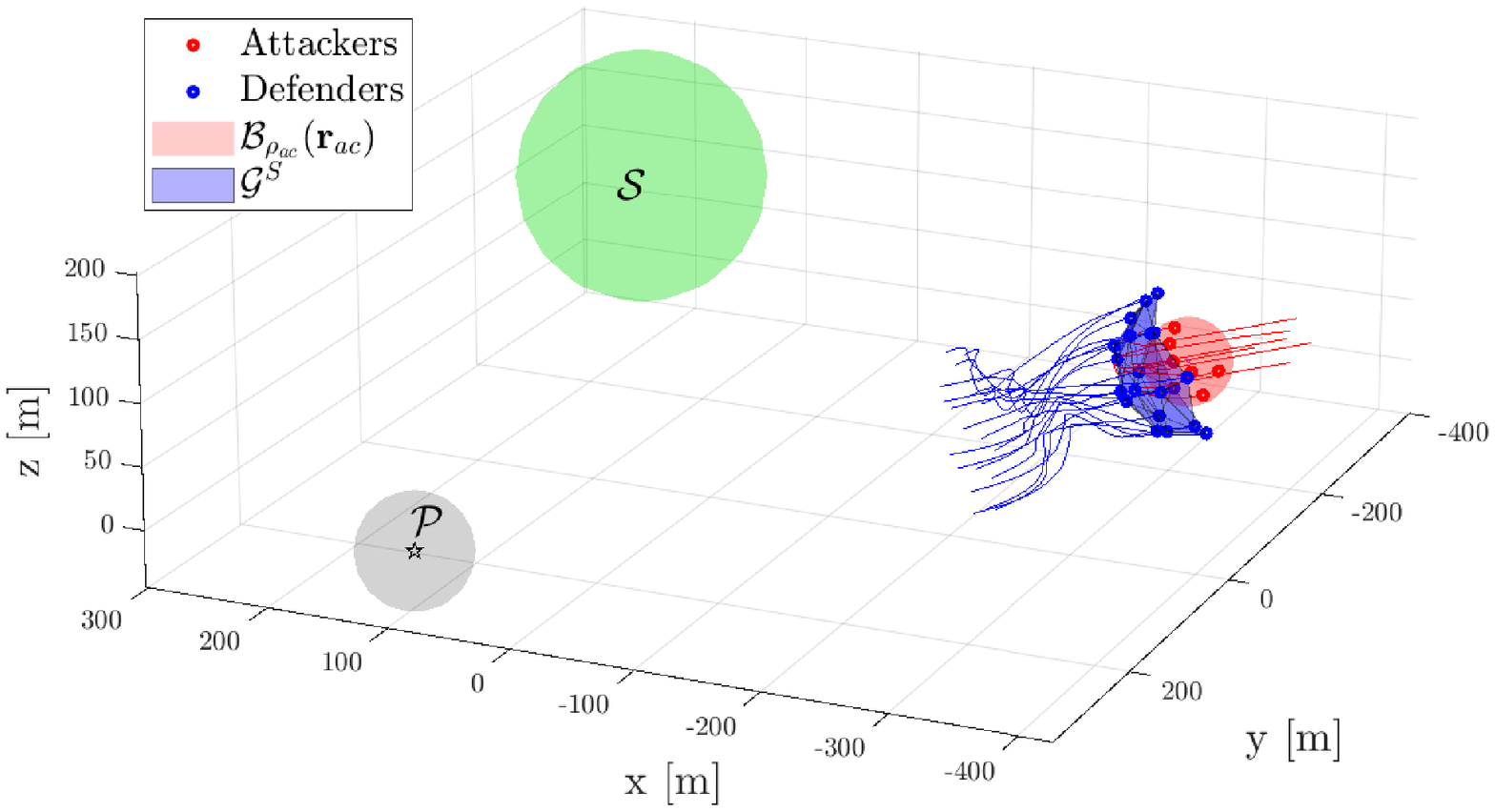}
			\caption{Seeking Phase: Hemispherical StringNet}
			\label{fig:3DSwarmHerd2}
		\end{subfigure}
		\begin{subfigure}[h]{0.48\textwidth}
			\centering
			\includegraphics[width=1\linewidth,trim={.9cm 2.3cm 2cm 3.0cm},clip]{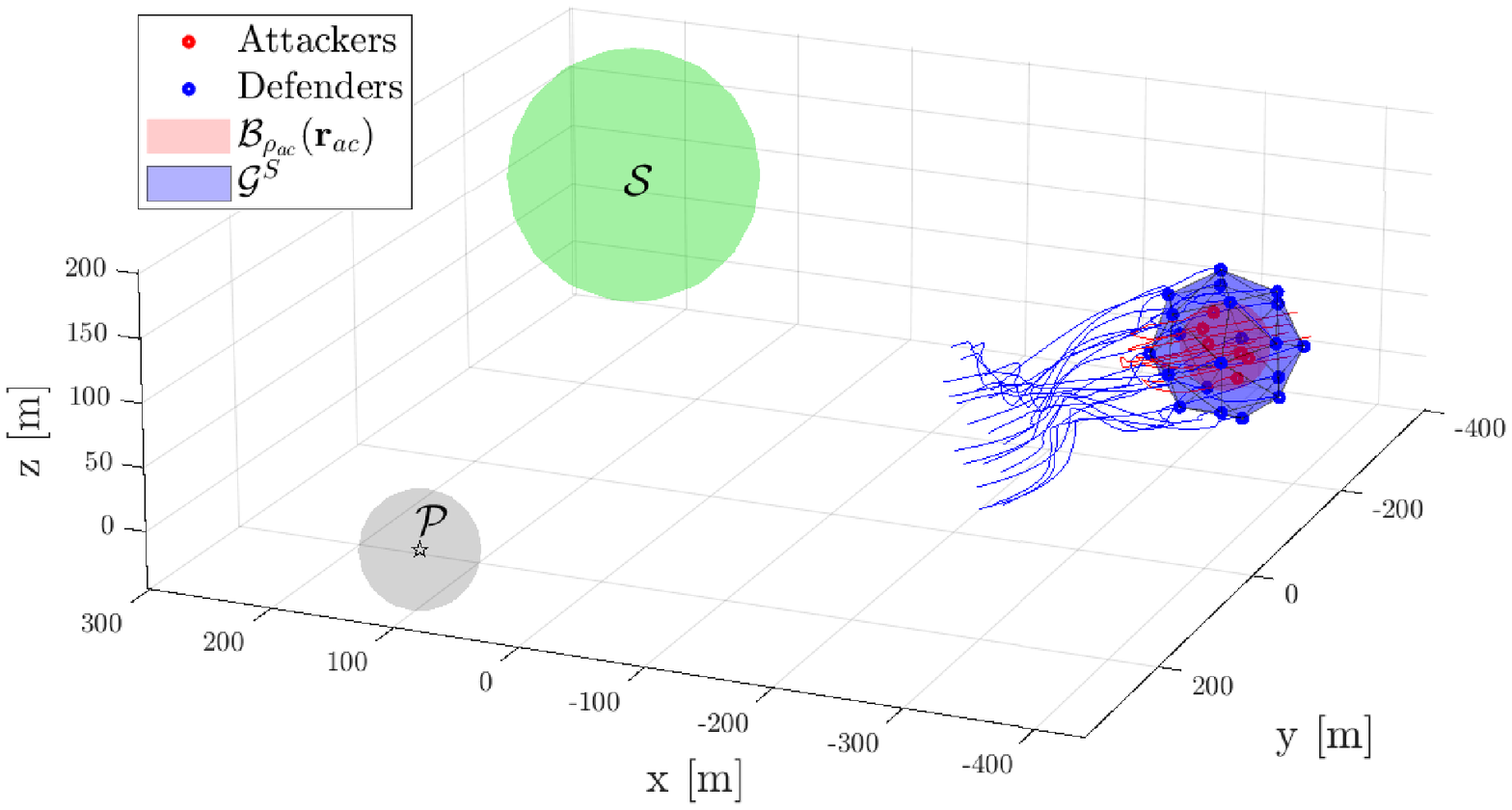}
			\caption{Enclosing Phase: Spherical StringNet}
			\label{fig:3DSwarmHerd3}
		\end{subfigure}
		\begin{subfigure}[h]{0.48\textwidth}
			\centering
			\includegraphics[width=1\linewidth,trim={.9cm 2.3cm 2cm 3.0cm},clip]{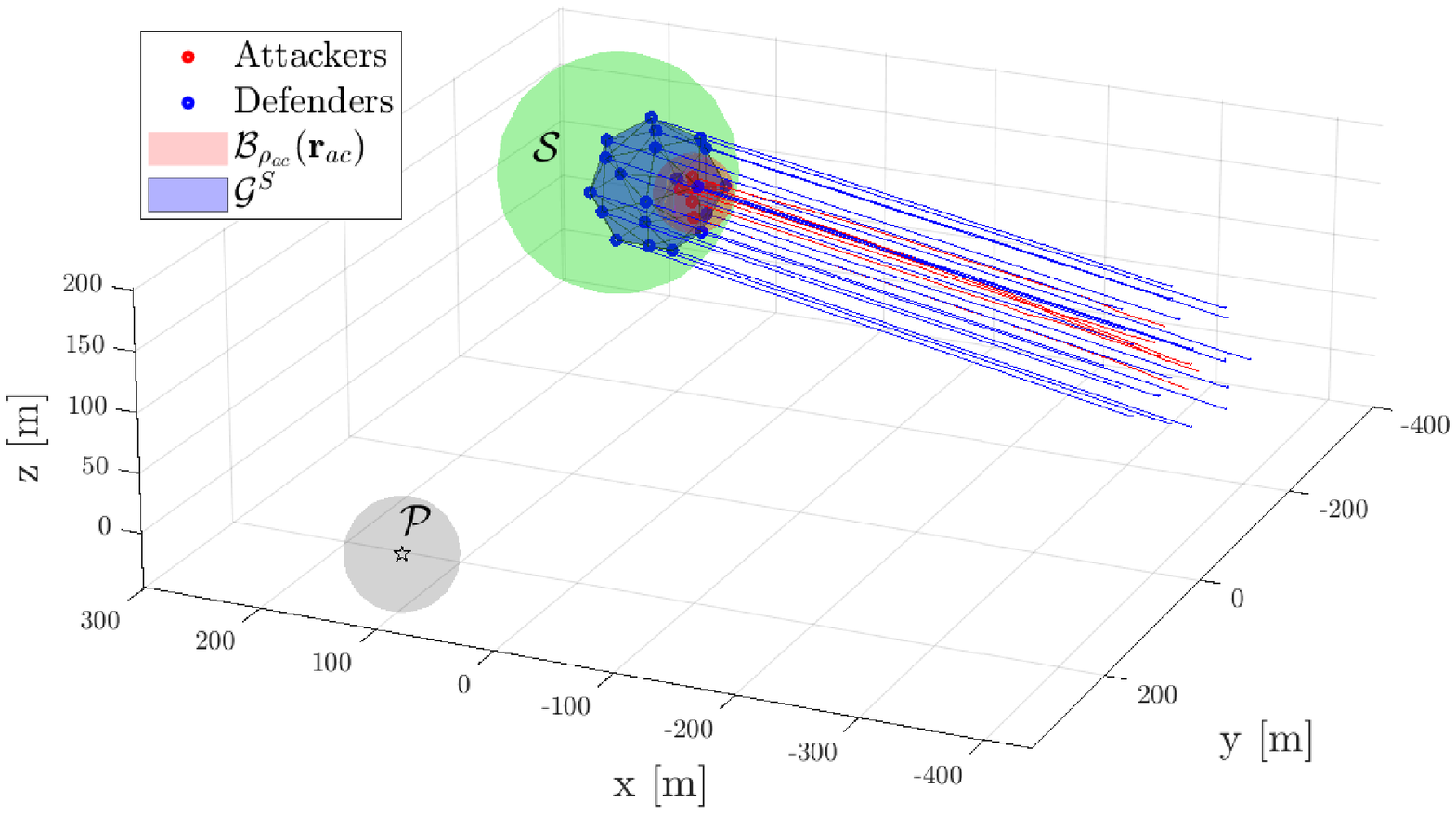}
			\caption{Complete Herding}
			\label{fig:3DSwarmHerd4}
		\end{subfigure}
		\caption{Snapshots of the paths of the agents during 3D-StringNet Herding}		\label{fig:multiSwarmHerd}
	\end{figure*}
	
\vspace{-1mm}
	\section{Conclusions} \label{sec:conclusions}
	We extended our 2D StringNet herding approach to 3D environments by defining the concept of 3D-StringNet. We designed three types of 3D-StringNet formations to capture and herd the attackers with the minimum number of defenders. The closed formation is a uniformly distributed spherical formation that can restrict the attackers' motion and herd them to the safe area. The other two formations: planar and hemispherical formation are generated from the spherical formation by using two carefully chosen mapping functions that respect the conditions on the edges in the formations. Appropriate modifications to the 2D herding control laws are provided for it to be applicable to 3D. The simulation shows the effectiveness of the proposed 3D-StringNet herding approach.
	
	Furthermore, we also provide a convex optimization formulation to quickly determine if a group of defenders starting at given positions can gather at a specified formation centered at a location on the shortest path of the attackers to the protected area before any attacker reaches the center of the formation. 
	 
	


	
	
	\bibliographystyle{IEEEtran}
	\bibliography{ACC2021_Refs}

\begin{thebibliography}{10}
\providecommand{\url}[1]{#1}
\csname url@samestyle\endcsname
\providecommand{\newblock}{\relax}
\providecommand{\bibinfo}[2]{#2}
\providecommand{\BIBentrySTDinterwordspacing}{\spaceskip=0pt\relax}
\providecommand{\BIBentryALTinterwordstretchfactor}{4}
\providecommand{\BIBentryALTinterwordspacing}{\spaceskip=\fontdimen2\font plus
\BIBentryALTinterwordstretchfactor\fontdimen3\font minus
  \fontdimen4\font\relax}
\providecommand{\BIBforeignlanguage}[2]{{%
\expandafter\ifx\csname l@#1\endcsname\relax
\typeout{** WARNING: IEEEtran.bst: No hyphenation pattern has been}%
\typeout{** loaded for the language `#1'. Using the pattern for}%
\typeout{** the default language instead.}%
\else
\language=\csname l@#1\endcsname
\fi
#2}}
\providecommand{\BIBdecl}{\relax}
\BIBdecl

\bibitem{bayindir2016review}
L.~Bay{\i}nd{\i}r, ``A review of swarm robotics tasks,'' \emph{Neurocomputing},
  vol. 172, pp. 292--321, 2016.

\bibitem{chipade2019swarmherding}
V.~S. Chipade and D.~Panagou, ``Herding an adversarial swarm in an obstacle
  environment,'' in \emph{2019 IEEE 58th Conference on Decision and Control
  (CDC)}.\hskip 1em plus 0.5em minus 0.4em\relax IEEE, 2019, pp. 3685--3690.

\bibitem{chipade2020swarmherding}
\BIBentryALTinterwordspacing
------, ``Multi-agent planning and control for swarm herding in 2d obstacle
  environments under bounded inputs,'' \emph{(Accepted in IEEE Transactions on
  Robotics) https://tinyurl.com/yy5k6943}, 2020. [Online]. Available:
  \url{https://tinyurl.com/yy5k6943}
\BIBentrySTDinterwordspacing

\bibitem{paranjape2018robotic}
A.~A. Paranjape, S.-J. Chung, K.~Kim, and D.~H. Shim, ``Robotic herding of a
  flock of birds using an unmanned aerial vehicle,'' \emph{IEEE Transactions on
  Robotics}, vol.~34, no.~4, pp. 901--915, 2018.

\bibitem{pierson2015bio}
A.~Pierson and M.~Schwager, ``Bio-inspired non-cooperative multi-robot
  herding,'' in \emph{IEEE International Conference on Robotics and
  Automation}, 2015, pp. 1843--1849.

\bibitem{varava2017herding}
A.~Varava, K.~Hang, D.~Kragic, and F.~T. Pokorny, ``Herding by caging: a
  topological approach towards guiding moving agents via mobile robots,'' in
  \emph{Robotics: Science and Systems}, 2017.

\bibitem{licitra2017single}
R.~A. Licitra, Z.~D. Hutcheson, E.~A. Doucette, and W.~E. Dixon, ``Single agent
  herding of n-agents: A switched systems approach,'' \emph{IFAC-PapersOnLine},
  vol.~50, no.~1, pp. 14\,374--14\,379, 2017.

\bibitem{kim2018three}
J.~Kim, ``Three-dimensional discrete-time controller to intercept a targeted
  uav using a capture net towed by multiple aerial robots,'' \emph{IET Radar,
  Sonar \& Navigation}, vol.~13, no.~5, pp. 682--688, 2018.

\bibitem{8360424}
Y.~{Jia}, Q.~{Li}, and S.~{Qiu}, ``Distributed leader-follower flight control
  for large-scale clusters of small unmanned aerial vehicles,'' \emph{IEEE
  Access}, vol.~6, pp. 32\,790--32\,799, 2018.

\bibitem{ritz2012cooperative}
R.~Ritz, M.~W. M{\"u}ller, M.~Hehn, and R.~D'Andrea, ``Cooperative quadrocopter
  ball throwing and catching,'' in \emph{2012 IEEE/RSJ International Conference
  on Intelligent Robots and Systems}.\hskip 1em plus 0.5em minus 0.4em\relax
  IEEE, 2012, pp. 4972--4978.

\bibitem{koay2014distributing}
C.~G. Koay, ``Distributing points uniformly on the unit sphere under a mirror
  reflection symmetry constraint,'' \emph{Journal of Computational Science},
  vol.~5, no.~5, pp. 696--700, 2014.

\bibitem{wie1989quaternion}
B.~Wie, H.~Weiss, and A.~Arapostathis, ``Quaternion feedback regulator for
  spacecraft eigenaxis rotations,'' \emph{Journal of Guidance, Control, and
  Dynamics}, vol.~12, no.~3, pp. 375--380, 1989.

\bibitem{stipanovic2012monotone}
D.~M. Stipanovi{\'c}, C.~J. Tomlin, and G.~Leitmann, ``Monotone approximations
  of minimum and maximum functions and multi-objective problems,''
  \emph{Applied Mathematics \& Optimization}, vol.~66, no.~3, pp. 455--473,
  2012.

\bibitem{chipade2020approximate}
V.~S. Chipade and D.~Panagou, ``Approximate time-optimal trajectories for
  damped double integrator in 2d obstacle environments under bounded inputs,''
  \emph{arXiv preprint arXiv:2007.05155}, 2020.

\bibitem{boyd2004convex}
S.~Boyd and L.~Vandenberghe, \emph{Convex optimization}.\hskip 1em plus 0.5em
  minus 0.4em\relax Cambridge university press, 2004.

\end{thebibliography}
\end{document}